\documentclass[letterpaper]{article}

\usepackage{authblk}

\usepackage{geometry}

\usepackage{xcolor}

\usepackage{amsmath,amsthm,amssymb,amsfonts}

\usepackage{mathtools}

\usepackage{proof}

\usepackage{comment}

\usepackage{listings}
\usepackage{tikz}
\usetikzlibrary{graphs}
\usetikzlibrary{trees,arrows}
\usetikzlibrary{graphs.standard}
\usetikzlibrary{shapes.geometric}
\usetikzlibrary{arrows.meta}
\tikzset{
  n/.style={shape=circle, draw=black, minimum size=20pt, font=\small},
  diredge/.style={->, >=stealth, thick},
  undiredge/.style={-, >=stealth, thick},
}

\usepackage[natbib=true, sorting=nyt, sortcites=true, maxcitenames=2]{biblatex}
\newtheorem*{theorem*}{Theorem}
\newtheorem{theorem}{Theorem}[section]

\newtheorem{proposition}[theorem]{Proposition}

\newtheorem{lemma}[theorem]{Lemma}

\usepackage{hyperref}

\let\temp\varphi
\let\varphi\phi
\let\phi\temp

\newcommand{\E}{\mathbb{E}}

\newcommand{\R}{\mathbb{R}}
\newcommand{\Z}{\mathbb{Z}}
\newcommand{\Q}{\mathbb{Q}}
\newcommand{\C}{\mathbb{C}}

\DeclarePairedDelimiter{\iop}{(}{)}
\DeclarePairedDelimiter{\iob}{[}{]}
\DeclarePairedDelimiter{\ios}{\{}{\}}

\newcommand{\p}{\iop*}
\newcommand{\s}{\ios*}
\renewcommand{\b}{\iob*}

\DeclareMathOperator{\NE}{NE}

\newcommand{\zo}{\s{0,1}}

\DeclareMathOperator{\vol}{Vol}
\DeclareMathOperator*{\mv}{MV}

\newcommand{\drn}{{\mathbin{!n}}}

\newcommand{\sdrn}{S_\drn}
\DeclareMathOperator*{\sys}{\mathcal{S}}
\DeclareMathOperator*{\game}{\mathcal{G}}
\DeclareMathOperator*{\gal}{Gal}

\newcommand{\algdeg}[1]{[\Q(#1):\Q]}

\begin{document}

\title{Nash Equilibria with Derangement Degree Probabilities}

\author{
    \begin{tabular}{c}
        Edan Orzech$^{1}$ \quad Martin Rinard$^{1,2}$\\
        $^{1}$MIT CSAIL\\
        $^{2}$National University of Singapore\\
        \vspace{0.5em}
        \texttt{\{iorzech,rinard\}@csail.mit.edu}
    \end{tabular}
}

\date{\today}

\maketitle

\begin{abstract}
    We prove for every $n\ge4$ the existence of an $n$-player game in normal form with integer payoffs that has a unique Nash equilibrium, which is fully mixed. In the equilibrium, each probability weight is an algebraic number of degree $\drn$ (the derangement number), and its minimal polynomial has Galois group $S_\drn$ and $\drn+1$ nonzero coefficients.
\end{abstract}

\section{Introduction}

In~\cite{orzech2025nash}, we presented for every $n\ge4$ an $n$-player $2$-action game $G_n$ with payoffs in $\s{0,1,2}$ and a unique Nash equilibrium (NE) $x^n=(x^n_1,\ldots,x^n_n)$ ($x^n_i\in[0,1]$ is the probability that player $i$ plays the first action) where every $x^n_i$ is irradical, meaning it is algebraic but not closed form expressible by radicals over $\Q$. There, $\algdeg{x^n_i}\in\s{6,26}$ and the Galois group (of the Galois closure of $\Q(x^n_i)/\Q$) is (isomorphic to) $S_6$ or $S_{26}$.

In this note we present a proof of a stronger result. The proof was developed using AI (see Section~\ref{sec:ai}). Let $\drn$ be the \emph{derangement number}: the number of permutations $n\to n$ without fixed points. It is known that $\drn=\b{\frac{n!}{e}}$. We prove the following:
\begin{proposition}\label{prop:main}
    For every $n\ge4$ there is an $n$-player $2$-action game $g^*_n$ with integer payoffs and a unique NE $y^n=(y^n_1,\ldots,y^n_n)$ such that for all $i\in[n]$:
    \begin{enumerate}
        \item $\algdeg{y^n_i}=\drn$,
        \item the Galois group of the Galois closure of $\Q(y^n_i)/\Q$ is $S_\drn$,
        \item the minimal polynomial $P^n_i$ of $y^n_i$ is dense: $P^n_i(t)=\sum_{k=0}^\drn b^n_{k,i}t^k$ where for every $k$, $b^n_{k,i}\ne0$.\footnote{In the proof we will fix $n$ and omit the $n$ superscript.}
    \end{enumerate}
\end{proposition}

We defer the question of upper bounds to a later version of this note.
$y^n_i$ is irradical for all $n\ge4$ and $i\in[n]$ because $S_\drn$ is an unsolvable group for all $n\ge4$.

As discussed in~\cite{orzech2025nash}, there are computational implications of this result.
To play the NE $y^n$, each player~$i$ must be able to sample an action from the underlying probability distribution $(y^n_i,1-y^n_i)$.
Since $y^n_i$ is irrational, performing exact sampling from $(y^n_i,1-y^n_i)$ appears to require generating a power series representation of $y^n_i$~\cite{flajolet2011buffon,mendo2020simulating}. One way to generate the series requires storing and evaluating a polynomial $f$ that satisfies $f(y^n_i)=0$. Since $\algdeg{y^n_i}=\drn$, $f$ would have $\deg f\ge\drn$ and typically take up $\Omega(\drn)$ space. Since $y^n_i$ is irradical, storing and evaluating such an $f$ appears to be necessary for player $i$ to play the only NE mixed strategy in the game $g^*_n$. Irradicality precludes the existence of a closed form representation in radicals for $y^n_i$, a representation that in general may be more compact compared to the minimal polynomial size (for example $a=\frac{1}{2}\sqrt{\frac{1}{2}+\sqrt{\frac{1}{2}+\sqrt{\frac{1}{2}}}}$ has the minimal polynomial $m_a=4096x^8-2048x^6+128x^4+32x^2-7$)~\cite{dixon1968solvable}.

Item $3$ of Proposition~\ref{prop:main} implies that the minimal polynomial of each $y^n_i$ requires $\Omega(\drn)$ space to store (without instance-specific optimizations), making the task of playing the only NE of $g^*_n$ computationally intensive for all players.\footnote{One could ask whether a minimal polynomial with many big coefficients has a multiple which has many
zero coefficients. However, \citet{Giesbrecht2012} suggest that a sparse multiple can have a much higher
degree, and its nonzero coefficients could be even larger.} Without item $3$, the minimal polynomial may be very easy to store and evaluate: Selmer's polynomial of degree $\drn$, $t^\drn-t-1$, is irreducible and has Galois group $S_\drn$ yet is sparse~\cite{selmer1956,osada1987galois}.

\section{Use of AI}\label{sec:ai}
We used ChatGPT 5.5 Pro to find this proof as follows. We gave the AI a prompt containing the statement in Proposition~\ref{prop:main} and asking it to prove the statement. It output a solution as a proof sketch, containing most of the technical tools and details required for the full proof. Afterward, we used further prompting to get more details about the solution. We verified its correctness and wrote up the full proof presented in this note. Below is a prompt that generated the solution.
\begin{lstlisting}
prove:

for every n>=4 there is an n-player 2-action game with integer payoffs and a unique NE y=(y_1,...,y_n) (y_i=Pr(i plays first action)) such that:
1. [Q(y_i):Q]=!n the derangement number
2. the galois group of y_i is S_!n
3. the minimal polynomial P_i of y_i has all its !n+1 coefficients nonzero
\end{lstlisting}
The generated solution output to the prompt is in Appendix~\ref{app:app} (lightly formatted with links removed). The proof we present is our formalized write-up based on that solution.

\section{Definitions}

We refer to~\cite{roughgarden2010algorithmic,orzech2025nash} for the full details regarding background on normal form games and NEs and~\cite{artin2011algebra,cox1997ideals,cox1998using} for a full algebra background. An $n$-player $2$-action game $G$ is defined by $n$ payoff functions, $u_1,\ldots,u_n:\zo^n\to\R$ which are extended to $[0,1]^n$ as $u_i(x)=\E_{a\sim x}(u_i(a))$. $G$ has \emph{integer payoffs} if $u_i(x_1,\ldots,x_n)\in\Z$ for every $i$ and $x_1,\ldots,x_n\in\zo$. A NE is defined in the usual way. Given $G$ we have $n$ multi-affine polynomials defined as $f_i(x_{-i})=u_i(x_{-i},x_i=1)-u_i(x_{-i},x_i=0)$. $f_i(x_{-i})$ or in short $f_i(x)$ is the payoff advantage the first action has over the second action for player~$i$ given that the other players play $x_{-i}$. We have the following useful lemma:

\begin{lemma}\label{lmm:ne-poly-cond}
  If $x\in\NE(G)$ then for every $i$,
  \begin{align}\label{eq:poly-cond}
    f_i(x_{-i})>0&~\Rightarrow~x_i=1,\nonumber\\
    f_i(x_{-i})<0&~\Rightarrow~x_i=0,\\
    0<x_i<1&~\Rightarrow~f_i(x_{-i})=0.\nonumber
  \end{align}
\end{lemma}

Here, we are interested in fully mixed NEs, in which $x_i\in(0,1)$ for all $i$. The set of fully mixed NEs of $G$ is exactly $\s{x\in(0,1)^n\mid\forall i\in[n]~~f_i(x)=0}$.

\paragraph{Games as equation systems:}
Define the generic system $\sys(c,x)$ as follows: there is a coefficient vector $c=(c_{i,s})_{i\in[n],~s\subseteq[n]\setminus\s{i}}$ which is either \emph{generic} (a vector of symbols) or \emph{in $\R^N$} where $N=n2^{n-1}$, a variable vector $x=(x_1,\ldots,x_n)$ that will take values in $\C^n$ and polynomials $f_1,\ldots,f_n\in\Q(c)[x]$ of the form
\begin{align*}
    f_i(c,x)=\sum_{s\subseteq[n]\setminus\s{i}}c_{i,s}\prod_{j\in s}x_j.
\end{align*}
When $c$ is generic, we may still set some $c_{i,s}$ to $0$.
We write $f_i^c\coloneqq f_i(c,\cdot)$.
Given $c$, each polynomial~$f_i^c$ has monomial exponents with support $S_i=\s{s\mid c_{i,s}\ne0}\subseteq\zo^{i-1}\times\s{0}\times\zo^{n-i}$. The monomial exponent supports $S_i$ of the generic $\sys(c,x)$ are $S_i=\zo^{i-1}\times\s{0}\times\zo^{n-i}$ for all $i$.

Every $c$ (either generic or real) defines a game $\game(c)$ uniquely up to affine transformations, which do not affect the game's best responses or NEs. In particular if $c\in\Q^N$ then the payoffs in $\game(c)$ can be scaled to be integers. Therefore we sometimes identify games as vectors in $\R^N$. Furthermore, the set of fully mixed NEs of $G(c)$ is $\s{x\in(0,1)^n\mid\sys(c,x)=0}$.

\subsection{Proof of Proposition~\ref{prop:main}}

The proof consists of $4$ steps. In the first, we show that the generic $\sys$ has $\drn$ solutions in $x\in(\C\setminus\s{0})^n$, and the field extension of $\Q(c)$ by all the coordinates of these $\drn$ roots has Galois group~$S_\drn$.

\begin{proposition}[Solution count of $\sys$]\label{prop:generic sol count}
    Suppose $c_{i,s}$ are generic. The number of solutions in $x\in(\C\setminus\s{0})^n$ to $\sys(c,x)=0$ is $\le\drn$, and $=\drn$ if $c_{i,s}\ne0$ for all $i,s$.
\end{proposition}
\begin{proof}
    Each $f_i$'s support of monomial exponents is $S_i\subseteq\zo^{i-1}\times\s{0}\times\zo^{n-i}$, and its associated Newton polytope is $\delta_i\subseteq\Delta_i\coloneqq[0,1]^{i-1}\times\s{0}\times[0,1]^{n-i}$.
    Therefore by the BKK theorem~\cite{bernshtein1975number,kouchnirenko1976polyedres} the number of solutions to $\sys$ (in the generic case) is $n!\mv(\delta_1,\ldots,\delta_n)$ where $\mv$ is the mixed volume function. By monotonicity, $n!\mv(\delta_1,\ldots,\delta_n)\le n!\mv(\Delta_1,\ldots,\Delta_n)$.

    To calculate the rhs we refer to the function
    \begin{align*}
        V(\lambda_1,\ldots,\lambda_n)\coloneqq\vol_n(\lambda_1\Delta_1+\ldots+\lambda_n\Delta_n)=\sum_{i_1,\ldots,i_n=1}^n\mv(\Delta_{i_1},\ldots,\Delta_{i_n})\prod_{l=1}^n\lambda_{i_l}.
    \end{align*}
    Let $\sum\lambda_{-i}=\sum_{j\ne i}\lambda_j$.
    \begin{align*}
        \vol_n(\lambda_1\Delta_1+\ldots+\lambda_n\Delta_n)=\vol_n\p{\b{0,\sum\lambda_{-1}}\times\ldots\times\b{0,\sum\lambda_{-n}}}=\prod_i\sum_{j\ne i}\lambda_j.
    \end{align*}
    The upper bound $n!\mv(\Delta_1,\ldots,\Delta_n)$ is the coefficient of $\prod_i\lambda_i$ in the rhs which is $\drn$.

    Equality is obtained if $\delta_i=\Delta_i$ for all $i$, so $c_{i,s}\ne0$ for all $i,s$.
\end{proof}

Let $r=\s{r^1,\ldots,r^\drn}$ be the set of solutions to $\sys$ where $r^k=(r^k_1,\ldots,r^k_n)$. We also denote them sometimes by $r(c)$.

\begin{proposition}[Galois group of generic roots]\label{prop:generic galois group}
    Suppose $c_{i,s}\ne0$ and generic for all $i,s$.
    \begin{align*}
        \gal(\Q(c)(r^1_1,\ldots,r^\drn_n)/\Q(c))\cong\sdrn.
    \end{align*}
\end{proposition}
\begin{proof}
    We use \cite[Theorem 1.5]{esterov2019galois}. First we need to prove that $S=(S_1,\ldots,S_n)$ is reduced and irreducible. We view each $S_i$ as a subset of $\Z^n$.

    $S$ is reduced if there are no $b_1,\ldots,b_n\in\Z^n$ such that $S_1+b_1,\ldots,S_n+b_n$ are all contained in a proper sublattice of $\Z^n$. Since $0,e_j\in S_i$ for all $i\ne j$, every sublattice that contains $S_1+b_1,\ldots,S_n+b_n$ contains $e_1,\ldots,e_n$ (by taking differences). Therefore $S$ is reduced.

    $S$ is irreducible if there are no $m>0$ and a subset $\s{S_{i_1},\ldots,S_{i_m}}\subseteq\s{S_1,\ldots,S_n}$ that can be shifted into a sublattice of $\Z^n$ of dimension $m$. In the first case, $m=1$. Every $S_i$ contains $0$ and $n-1$ of the $n$ standard basis elements of $\Z^n$, so the shifted $S_i$ can only fit into a sublattice of dimension at least $n-1>1$. In the second case, $m>1$. The union of every pair $S_i,S_j$ ($i\ne j$) contains $0$ and all of $e_1,\ldots,e_n$. Therefore the shifted $S_{i_1},\ldots,S_{i_m}$ can only fit into a sublattice of dimension $n$ (by considering differences of elements in these tuples), namely $\Z^n$. Therefore $S$ is irreducible.

    Therefore, from \cite{esterov2019galois} the monodromy group of $\s{r^k}_k$ is $G\cong S_{n!\mv(\Delta_1,\ldots,\Delta_n)}=\sdrn$ (last equality by Proposition~\ref{prop:generic sol count}).

    Lastly we prove that $G\cong\gal(\Q(c)(r^1_1,\ldots,r^\drn_n)/\Q(c))$. Write in short $\Q(c,r)$ and $\C(c,r)$. First observe that $\Q(c,r)/\Q(c)$ is a Galois extension, because $r=\s{r^k_i}_{k,i}$ contains all the conjugates of the $r^k_i$ by definition (as the set/tuple of all solution coordinates to $\sys$). Now, using~\cite{brysiewicz2021solving}, $G\cong\gal(\C(c,r)/\C(c))$. 
    Let $\sigma\in\gal(\C(c,r)/\C(c))$. Then $\sigma\restriction_{\Q(c)}=id_{\Q(c)}$, $\sigma(\Q(c,r))=\Q(c,r)$ and $\sigma\restriction_{\Q(c,r)}\in\gal(\Q(c,r)/\Q(c))$ because $\sigma$ permutes the root tuples $r^1,\ldots,r^\drn$. Also, if $\sigma,\tau\in\gal(\C(c,r)/\C(c))$ and $\sigma\restriction_{\Q(c,r)}=\tau\restriction_{\Q(c,r)}$ then $\sigma,\tau$ agree on $r$. They both fix $\C(c)$ so we get $\sigma=\tau$. Therefore $\gal(\C(c,r)/\C(c))\lesssim\gal(\Q(c,r)/\Q(c))$. Overall:
    \begin{align*}
        S_\drn\cong G\lesssim\gal(\Q(c,r)/\Q(c))\lesssim S_\drn.
    \end{align*}
    So $\gal(\Q(c,r)/\Q(c))\cong\sdrn$.
\end{proof}

In the second step, we create a simple game $\game(c_0)$ with a unique NE $\p{\frac12,\ldots,\frac12}$. Then we find an open neighborhood $c_0\in U\subseteq\R^N$ such that every $\game(c)$ with $c\in U$ also has a unique NE.

\begin{proposition}[Simple anchor game with open neighborhood of unique fully mixed NE games]
    There is an $n$-player 2-action game $g_0=\game(c_0)$ with a unique NE $\p{\frac12,\ldots,\frac12}$ and an open set $U\subseteq\R^N$ such that $c_0\in U$ and for every $c\in U$, $\game(c)$ has a unique NE.
\end{proposition}
\begin{proof}
    Define the game through the polynomials $f_i(x)=2x_{i+1}-1$ for $i<n$ and $f_n(x)=1-2x_1$. Let $x\in\NE$. If $x_{i+1}=0$ then $f_i=-1<0$ so $x_i=0$ so we get $x_1=\ldots=x_{i+1}=0$ and then $f_n=1>0$ so $x_n=1$ so $f_{n-1}=1>0$ so $x_{n-1}=1$. Therefore we also get $x_{i+1}=\ldots=x_n=1$ -- a contradiction. Similarly we get a contradiction if $x_{i+1}=1$ or $x_1\in\zo$. Therefore all NEs are fully mixed and the unique solution to $f_1=\ldots=f_n=0$ is $x^0_1=\ldots=x^0_n=\frac12$.

    Now we prove the existence of $U$.
    Define $F:\R^{N+n}\to\R^n$ by $F(c,x)=(f_1^c(x),\ldots,f_n^c(x))$. The matrix $\b{\frac{\partial f_i}{\partial x_j}(g_0,x^0)}_{i,j}$ is invertible:
    \begin{align*}
        \frac{\partial f_i}{\partial x_j}(g_0,x^0)=\frac{\partial f^{c_0}_i}{\partial x_j}\p{\frac12,\ldots,\frac12}=\begin{cases}
            2_{j=i+1}&i<n\\
            -2_{j=1}&i=n
        \end{cases}
    \end{align*}
    meaning the matrix is
    \begin{align*}
        \left(\begin{array}{c|c}
            \bf0&2I_{n-1}\\\hline
            -2&\bf0
        \end{array}\right)
    \end{align*}
    which is invertible. The implicit function theorem gives open neighborhoods $U\ni c_0$ in $\R^N$ and $W\ni x^0$ in $\R^n$, and a unique function $\phi:U\to W$ such that for all $c\in U$, $\phi(c)$ is the unique solution in $W$ to $f_1^c(x)=\ldots=f_n^c(x)=0$. This means that $\phi(c)$ is a fully mixed NE (in $W$) of $\game(c)$. Wlog take $W\subseteq(0,1)^n$.

    Now we show that wlog $U$ contains only games with a unique and fully mixed NE. Suppose otherwise. Then there is $c_m\to c_0$ and $x^m\in\NE(\game(c_m))$ such that $x^m\ne\phi(c_m)$. Since $x^m\in[0,1]^n$, there is a subsequence $x^{m_k}\to\Tilde x$. The NE inequalities are closed, and $(c,x)\to f(c,x)$ is continuous in $c,x$, so $\Tilde x\in\NE(g_0)$. Since $g_0$ has a unique NE, $\Tilde x=x^0$ so for all large enough $k$, $x^{m_k}\in W$ and $x^{m_k}$ is fully mixed, meaning $f_1^{c_{m_k}}(x^{m_k})=\ldots=f_n^{c_{m_k}}(x^{m_k})=0$. By the uniqueness of $\phi$ in $W$, $x^{m_k}=\phi(c_{m_k})$ -- a contradiction. Therefore after shrinking $U$, every $\game(c)$ with $c\in U$ has a unique NE, and it is fully mixed.
\end{proof}

Let $P_i(t)=\prod_{k=1}^\drn(t-r^k_i)$. Then there are polynomials $b_{k,i}(r^1_i,\ldots,r^\drn_i)=b_{k,i}(r(c))$ such that $P_i=\sum_{k=0}^\drn b_{k,i}t^k$. In the third step we show that $P_i$:
\begin{enumerate}
    \item is a degree-$\drn$ polynomial in $\Q(c)[t]$ with $b_{k,i}\ne0$ for all $k$ (as functions of $c$),
    \item is the minimal polynomial of $r^k_i$ for all $k$,
    \item has Galois group $S_\drn$.
\end{enumerate}

\begin{proposition}\label{prop:Pi in Qct}
    Suppose $c_{i,s}\ne0$ and generic for all $i,s$. Then $P_i\in\Q(c)[t]$.
\end{proposition}
\begin{proof}
    First observe that $b_{k,i}(r^1_i,\ldots,r^\drn_i)$ is symmetric in $r^1_i,\ldots,r^\drn_i$. $\Q(c)$ is a perfect field so we have the Galois group $G=\gal(\overline{\Q(c)}/\Q(c))$ where $\overline{\Q(c)}$ is the algebraic closure of $\Q(c)$. By definition $\sigma(r^k)=(\sigma(r^k_1),\ldots,\sigma(r^k_n))$. Let $\sigma\in G$. Since $r^k_i\in\overline{\Q(c)}$ we have for all $i,j$ $\sigma(f_i^c(r^k))=f_i^c(\sigma(r^k))=0$ (note the absence of the variables $x_1,\ldots,x_n$ from both sides), so $\sigma$ permutes $\s{r^1,\ldots,r^\drn}$ within itself ($\sigma^{-1}(0)=0$ so $\sigma$ does not send any $r^i_j\ne0$ to $0$). Therefore $\sigma$ also permutes each $\s{r^1_i,\ldots,r^\drn_i}$ within itself. Therefore $\sigma(b_{k,i})=b_{k,i}$. So $b_{k,i}$ is a fixed point of all $\sigma\in G$, so $b_{k,i}\in\Q(c)$. Therefore $P_i\in\Q(c)[t]$.
\end{proof}

Let $\lambda\in\Q^n$ and consider $\sys(c,x+\lambda)$. Then there is an invertible function $\psi_\lambda$ such that $\sys(c,x+\lambda)=\sys(\psi_\lambda(c),x)$. In particular every $f^{\psi_\lambda(c)}_i$ is still multi-affine and does not depend on $x_i$, meaning $\game(\psi_\lambda(c))$ is a valid game.

\begin{proposition}[Dense generic $P_i$]\label{prop:dense}
    For every $k,i$, the polynomials $b_{k,i}(r)\in\Q(c)[r]$ are not $0$.
\end{proposition}
\begin{proof}
    Use the notation $P^c_i$ and consider $\sys(c,x+\lambda)$. Then it also has $\drn$ roots and these are $\s{r^k-\lambda}_k$. Since $\sys(c,x+\lambda)=\sys(\psi_\lambda(c),x)$ we get $P^{\psi_\lambda(c)}_i=\prod_k(t-(r^k_i-\lambda_i))=\sum_{k=0}^\drn b_{k,i}(r(\psi_\lambda(c)))t^k$. Also, $P^{\psi_\lambda(c)}_i(t)=P^c_i(t+\lambda_i)=\sum_{k=0}^\drn b_{k,i}(r(c))(t+\lambda_i)^k$. $P^c_i$ is monic so $b_{\drn,i}(r(c))=1$ and for all $k$:
    \begin{align*}
        b_{k,i}(r(\psi_\lambda(c)))=\sum_{j=k}^\drn\binom{j}{k}\lambda_i^{j-k}b_{j,i}(r(c))=\binom{\drn}{k}\lambda_i^{\drn-k}+\sum_{j=k}^{\drn-1}\binom{j}{k}\lambda_i^{j-k}b_{j,i}(r(c)).
    \end{align*}
    To see that $b_{k,i}(r(\psi_\lambda(c)))\ne0$ for some $\lambda$, treat the rhs as a polynomial $Q_k(\lambda_i)$. Its coefficients do not depend on $\lambda_i$. Therefore $Q_k\ne0$ so it has finitely many zeros, so there is a $\lambda$ such that for all~$k$, $b_{k,i}(r(\psi_\lambda(c)))=Q_k(\lambda_i)\ne0$. Therefore the polynomial $b_{k,i}$ is nonzero for all $k,i$.
\end{proof}

Now we can prove the second and third items.

\begin{proposition}
    $P_i$ is the minimal polynomial of $r^1_i,\ldots,r^\drn_i$ and its Galois group is $S_\drn$.
\end{proposition}
\begin{proof}
    First we prove that for every $k_1\ne k_2$ and for every $i$, $r^{k_1}_i\ne r^{k_2}_i$.
    Suppose this is not the case, then for every $\sigma\in S_\drn$ (taken from the Galois group in Proposition~\ref{prop:generic galois group}) such that $\sigma(k_1)=l_1,\sigma(k_2)=l_2$ for any $l_1\ne l_2$ we get $\sigma(r^{k_1}_i)=r^{\sigma(k_1)}_i=r^{\sigma(k_2)}_i=\sigma(r^{k_2}_i)$. $S_\drn$ is $2$-transitive so $r^1_i=\ldots=r^\drn_i$.
    
    So $\s{r^k_i}_{k=1}^\drn$ are either all identical or all different. For some $i_0$, $\s{r^k_{i_0}}_{k=1}^\drn$ are all different, otherwise $r^1=\ldots=r^\drn$ meaning $f_1=\ldots=f_n=0$ has only one solution, in contradiction to having $\drn>1$ solutions.
    
    Now we permute $\sys$ by relabeling players by swapping $i,i_0$ in the coefficients, variables and functions that define $\sys$. This relabeled system is equivalent to $\sys(c,x)$ because all the coefficients and variables are generic. The relabeled system has solutions $\{\Tilde r^k\}_{k\in[\drn]}$. Applying the argument above to the relabeled system implies that $\s{\Tilde r^k_{i_0}}_k=\s{r^k_i}_k$ are all different. Therefore the only possible case is that for all $i$, $\s{r^k_i}_k$ are all different.
    
    Therefore $P_i$ is a separable polynomial and it is in $\Q(c)[t]$ (Proposition~\ref{prop:Pi in Qct}) so we have its Galois group $H=\gal(\Q(c)(r^1_i,\ldots,r^\drn_i)/\Q(c))$. By Proposition~\ref{prop:generic galois group} we also have the Galois group $G=\gal(\Q(c)(r^1_1,\ldots,r^\drn_n)/\Q(c))\cong S_\drn$. Define the homomorphism $\phi:G\to H$ by $\phi(\sigma)(r^k_i)=r^{\sigma(k)}_i$. Since $\s{r^k_i}_k$ are all different, $\phi$ is injective. Therefore $G\lesssim H$. Since $H\lesssim S_\drn$, we have $H\cong S_\drn$. Therefore $H$ is transitive so $P_i$ is irreducible.
    
    Overall $P_i$ is the minimal polynomial of $r^1_i,\ldots,r^\drn_i$ and its Galois group is $S_\drn$. 
\end{proof}

In the last step we finish the proof by specializing the generic $c$ in a way that preserves all the properties we proved in the previous steps:

\begin{proposition}[Specializing the generic $c$]
    There is a $c^*\in U\cap\Q^N$ and a game $g^*=\game(c^*)$ with integer payoffs such that:
    \begin{enumerate}
        \item $\NE(g^*)=\s{x^*}$,
        \item all the coordinates of $x^*$ are algebraic numbers of degree $\drn$ and Galois group $S_\drn$,
        \item the minimal polynomial of each $x^*_i$ is $P_i=\sum_{k=0}^\drn b_{k,i}(c^*)t^k$ where $b_{k,i}(c^*)\ne0$ for all $k,i$.
    \end{enumerate}
\end{proposition}
\begin{proof}
    Every game in $U$ has a unique NE. To satisfy (2), we want to specialize $f_1,\ldots,f_n$ correctly, maintaining the Galois groups of the solutions and the density of their minimal polynomials.

    For every $i$ use \cite[Proposition 3.3.5]{serre2016topics} and find a thin set $A_i\subseteq\Q^N$ such that for all $c\in\Q^N\setminus A_i$, $P_i$ is irreducible over $\Q$ and its Galois group is (isomorphic to) $S_\drn$. To $A_i$ add all the rational points at which $b_{k,i}$ has a zero or a pole. By Proposition~\ref{prop:dense}, zeros and poles of the $b_{k,i}$s add to $A_i$ a proper Zariski-closed set. Also add to $A_i$ the rational points at which $P_i$ becomes inseparable. $P_i$ is inseparable at rational points where its discriminant is zero. $P_i$ (with generic $c$) is separable so its discriminant's zero set is a proper Zariski-closed set. Overall $A_i$ is still thin. Let $A=\bigcup_iA_i$. Then $A$ is also thin by definition. For every $c\in\Q^N\setminus A$, every $P_i$ is irreducible over $\Q$, has Galois group isomorphic to~$S_\drn$ and all its $\drn+1$ coefficients are nonzero.

    Set $U$ to $U\cap\s{c\in\R^N\mid\forall i,s,~c_{i,s}\ne0}$. Note that $U$ is still open and nonempty. We now prove that $U\cap(\Q^N\setminus A)\ne\emptyset$. For an element $y\in Y$ and product set $\prod_{i\in I}Y_i$ where $Y,Y_i$ have the same number of coordinates, say that $y\in\prod_iY_i$ if $(y,\ldots,y)$ ($|I|$ times) is in the product set. Interpret $Y\subseteq\prod_iY_i$ similarly.
    
    We use \cite[Proposition 3.5.3]{serre2016topics}. Use their notation and set $S_0=\s{\infty}$. Then there is a finite set of primes $S$ and an open set $W\subseteq\prod_{p\in S}\Q_p^N$ such that $A\cap W=\emptyset$. For every prime $p$, $\Q^N$ is dense in $\Q_p^N$, so $W\cap\Q^N\ne\emptyset$, meaning that $\Q^N\setminus A\ne\emptyset$.
    
    Furthermore, $\Q^N$ has the weak approximation property \citep[Chapter 5]{serre2016topics}, i.e., $\Q^N$ is dense in $\R^N\times\prod_{p\in S}\Q_p^N$ with the suitable product topology.
    Therefore by definition of denseness there is~a~$c^*\in\Q^N$ such that $(c^*,c^*)\in U\times W$ so $c^*\in U\cap(\Q^N\setminus A)$.

    $g^*=\game(c^*)$ has integer payoffs ($c^*\in\Q^N$ and after multiplying the payoffs in $g^*$ by a common denominator) and a unique NE $x^*=(x^*_1,\ldots,x^*_n)$ ($c^*\in U$). Also $\algdeg{x^*_i}=\drn$ for all $i$ and the minimal polynomial $P_i$ of $x^*_i$ has all its $\drn+1$ coefficients nonzero ($c^*\notin A$), concluding the proof.
\end{proof}

This concludes the proof of Proposition~\ref{prop:main}.

\printbibliography

\appendix

\section{Generated Solution Output}\label{app:app}

Let \(D={!n}\), the number of derangements of \(\{1,\dots,n\}\). I will prove existence.

I will use two standard algebraic facts. First, the Bernstein--Kouchnirenko theorem: a generic sparse square system has as many torus roots as the lattice mixed volume of its Newton polytopes. Second, Esterov’s theorem: a reduced irreducible sparse square system has full symmetric monodromy/Galois group. Esterov’s paper states both the Kouchnirenko--Bernstein root count setup and the full symmetric monodromy theorem for reduced irreducible tuples. Hilbert irreducibility is then used to specialize the generic rational-coefficient system to one with rational coefficients while preserving the Galois groups; this is the standard Noether-Hilbert specialization principle.

\subsection*{1. Encoding two-action games by multiaffine equations}

Let \(x_i\in[0,1]\) denote the probability that player \(i\) plays action \(1\). In any two-action game, the expected payoff difference

\[
f_i(x_{-i})
=
\mathbb E[u_i(1,\cdot)]-\mathbb E[u_i(0,\cdot)]
\]

is a multiaffine polynomial in the variables \(x_j\), \(j\ne i\). Conversely, every such multiaffine polynomial can be realized by a two-action game: set the payoff difference at each pure opponents’ profile equal to the value of the polynomial at the corresponding cube vertex, and interpolate multilinearly.

A profile \(x\) is a Nash equilibrium exactly when, for every \(i\),

\[
x_i f_i(x_{-i})\ge 0,
\qquad
(1-x_i)f_i(x_{-i})\le 0.
\]

In particular, if \(x\) is totally mixed, then \(x\) is a Nash equilibrium iff

\[
f_1(x)=\cdots=f_n(x)=0.
\]

Now define the full two-action support

\[
A_i=\{\alpha\in\{0,1\}^n:\alpha_i=0\}.
\]

Thus \(f_i\) is a general polynomial supported on \(A_i\), i.e. it contains all squarefree monomials in the variables other than \(x_i\).

\subsection*{2. The mixed volume is the derangement number}

Let \(P_i=\operatorname{conv}(A_i)\). Then

\[
P_i=\{x\in\mathbb R^n:0\le x_j\le 1\text{ for }j\ne i,\ x_i=0\}.
\]

For nonnegative \(\lambda_1,\dots,\lambda_n\),

\[
\lambda_1P_1+\cdots+\lambda_nP_n
=
\prod_{j=1}^n
\left[0,\sum_{i\ne j}\lambda_i\right].
\]

Therefore its volume is

\[
\prod_{j=1}^n\left(\sum_{i\ne j}\lambda_i\right).
\]

The coefficient of \(\lambda_1\cdots \lambda_n\) is the number of ways to choose, from the \(j\)-th factor, a \(\lambda_i\) with \(i\ne j\), using every \(i\) exactly once. That is precisely the number of derangements of \(n\) letters. Hence

\[
\operatorname{MV}(P_1,\dots,P_n)={!n}=D.
\]

So a generic system

\[
f_1=\cdots=f_n=0
\]

with these supports has \(D\) isolated roots in \((\mathbb C^*)^n\).

\subsection*{3. The generic Galois group is \(S_D\), and each coordinate has degree \(D\)}

The tuple \((A_1,\dots,A_n)\) is reduced: the lattice generated by all differences \(A_i-A_i\) contains every basis vector \(e_j\), because \(e_j\in A_i-A_i\) whenever \(i\ne j\).

It is also irreducible. Indeed, if \(I\subset\{1,\dots,n\}\) has \(|I|=m<n\), then:

\[
\dim\sum_{i\in I}P_i
=
\begin{cases}
n-1,&m=1,\\
n,&m\ge 2.
\end{cases}
\]

Since \(n\ge4\), this dimension is always \(>m\). Thus Esterov’s theorem applies, and the generic monodromy group of the \(D\) roots is \(S_D\).

We also need the \(i\)-th coordinate itself to have degree \(D\). For a fixed coordinate \(k\), consider the equivalence relation on the \(D\) generic roots:

\[
r\sim s
\quad\Longleftrightarrow\quad
(x_k\text{ on root }r)=(x_k\text{ on root }s).
\]

This relation is invariant under monodromy. Since the monodromy group is \(S_D\), the relation is either discrete or universal. It is not universal: choose constants \(b_{ij}\) with \(i\ne j\), distinct in each column \(j\), and consider

\[
h_i(x)=\prod_{j\ne i}(x_j-b_{ij}).
\]

The roots of \(h_1=\cdots=h_n=0\) are exactly the points indexed by derangements \(\sigma\), namely

\[
x_{\sigma(i)}=b_{i,\sigma(i)}.
\]

For each coordinate \(k\), different derangements can have different \(\sigma^{-1}(k)\), and hence different \(x_k\). Therefore the generic relation is not universal, so it is discrete. Thus the generic \(k\)-th coordinate takes \(D\) distinct values.

Therefore, for every \(k\), the coordinate eliminant

\[
E_k(T)=\prod_{\text{generic roots }z}(T-z_k)
\]

has degree \(D\) and Galois group \(S_D\) over the rational function field of the coefficients.

\subsection*{4. A unique-equilibrium base game}

Now construct a simple base game with a unique equilibrium. Define

\[
f_i^0=x_{i+1}-\frac12\quad(1\le i<n),
\qquad
f_n^0=\frac12-x_1.
\]

At any equilibrium, for \(i<n\),

\[
x_{i+1}>\frac12\Rightarrow x_i=1,
\qquad
x_{i+1}<\frac12\Rightarrow x_i=0,
\]

while for player \(n\),

\[
x_1>\frac12\Rightarrow x_n=0,
\qquad
x_1<\frac12\Rightarrow x_n=1.
\]

If some coordinate is \(>1/2\), these implications propagate around the cycle and force the same coordinate to be \(<1/2\), a contradiction. The same argument rules out a coordinate \(<1/2\). Hence every equilibrium must be

\[
x^0=\left(\frac12,\dots,\frac12\right).
\]

At \(x^0\), all players are indifferent, so \(x^0\) is indeed an equilibrium. Thus the base game has a unique Nash equilibrium.

Moreover, this property is open. The Jacobian of

\[
(f_1^0,\dots,f_n^0)
\]

at \(x^0\) is a signed cyclic permutation matrix, hence invertible. By the implicit function theorem, the totally mixed zero near \(x^0\) persists uniquely under small coefficient perturbations. Compactness of \([0,1]^n\), together with the Nash inequalities above, rules out new boundary equilibria for all sufficiently small perturbations. Therefore there is a Euclidean open neighborhood \(V\) of \(f^0\) in coefficient space such that every system in \(V\) defines a game with a unique, totally mixed Nash equilibrium.

\subsection*{5. Specialize to rational coefficients with full Galois group}

The generic coordinate eliminants \(E_k(T)\) have Galois group \(S_D\). By Hilbert irreducibility, in its standard form allowing approximation inside a prescribed real open set, there exists a rational coefficient vector

\[
f=(f_1,\dots,f_n)\in V
\]

such that, for every \(k\), the specialized eliminant \(E_{k,f}(T)\) remains irreducible of degree \(D\) and has Galois group \(S_D\) over \(\mathbb Q\).

Let \(z=(z_1,\dots,z_n)\) be the unique Nash equilibrium of this rational game. Since it is totally mixed, it is one of the \(D\) complex roots of

\[
f_1=\cdots=f_n=0.
\]

For every \(k\), \(z_k\) is a root of \(E_{k,f}\). Since \(E_{k,f}\) is irreducible of degree \(D\), it is the minimal polynomial of \(z_k\). Hence

\[
[\mathbb Q(z_k):\mathbb Q]=D={!n},
\]

and the Galois group of \(z_k\) is \(S_D=S_{!n}\).

\subsection*{6. Make every coefficient of every minimal polynomial nonzero}

Let \(q_k(T)\) be the minimal polynomial of \(z_k\). Choose a rational number \(\varepsilon_k\) arbitrarily small. The minimal polynomial of

\[
y_k=z_k+\varepsilon_k
\]

is

\[
q_k(T-\varepsilon_k).
\]

For a fixed nonconstant polynomial \(q_k\), each coefficient of \(q_k(T-\varepsilon)\) is a polynomial function of \(\varepsilon\). The coefficient of \(T^{D-r}\) has highest \(\varepsilon\)-term

\[
\binom Dr(-\varepsilon)^r,
\]

so it is not identically zero. Therefore only finitely many rational \(\varepsilon\) make any coefficient vanish. Choose each \(\varepsilon_k\) small enough, rational, and outside this finite bad set.

Now define a shifted game by

\[
g_i(y_{-i})
=
f_i(y_1-\varepsilon_1,\dots,y_{i-1}-\varepsilon_{i-1},
y_{i+1}-\varepsilon_{i+1},\dots,y_n-\varepsilon_n).
\]

For sufficiently small \(\varepsilon_k\), the new coefficient vector remains inside \(V\), so the shifted game still has a unique Nash equilibrium. That equilibrium is

\[
y=(z_1+\varepsilon_1,\dots,z_n+\varepsilon_n).
\]

For every \(k\),

\[
[\mathbb Q(y_k):\mathbb Q]=D={!n},
\]

the Galois group is still \(S_D\), and the minimal polynomial

\[
P_k(T)=q_k(T-\varepsilon_k)
\]

has all \(D+1={!n}+1\) coefficients nonzero.

\subsection*{7. Clear denominators to get integer payoffs}

The shifted polynomials \(g_i\) have rational coefficients. Choose a positive integer \(M\) clearing all denominators of the vertex values \(g_i(v_{-i})\), where \(v_{-i}\in\{0,1\}^{n-1}\).

Define the actual integer-payoff game by

\[
u_i(1,v_{-i})=M g_i(v_{-i}),
\qquad
u_i(0,v_{-i})=0.
\]

Then the expected payoff difference for player \(i\) is exactly \(M g_i(y_{-i})\). Multiplication by the positive number \(M\) does not change best responses or Nash equilibria. Thus the resulting game has integer payoffs and the same unique Nash equilibrium \(y\).

Therefore, for every \(n\ge4\), there exists an \(n\)-player two-action integer-payoff game with a unique Nash equilibrium \(y\) such that, for every \(i\),

\[
[\mathbb Q(y_i):\mathbb Q]={!n},
\qquad
\operatorname{Gal}(P_i)=S_{!n},
\]

and the minimal polynomial \(P_i\) of \(y_i\) has all \({!n}+1\) coefficients nonzero.

\end{document}